\documentclass[11pt]{article}

\usepackage{amssymb,amsthm,amsmath}
\usepackage{graphicx}
\usepackage[small]{caption}
\usepackage[utf8]{inputenc}
\usepackage{amsfonts}
\usepackage{fullpage}
\usepackage{enumerate}
\usepackage{enumitem}
\usepackage{url}

\newtheorem{theorem}{Theorem}
\newtheorem{lemma}{Lemma}

\newcommand{\eps}{\varepsilon}

\def\A{\mathcal A}
\def\B{\mathcal B}
\def\C{\mathcal C}
\def\H{\mathcal H}
\def\I{\mathcal I}

\newcommand{\etal}{{et~al.}}
\newcommand{\ie}{{i.e.}}
\newcommand{\eg}{{e.g.}}

\newcommand{\later}[1]{}
\newcommand{\old}[1]{}

\title{\textsc{A Selectable Sloppy Heap}}

\author{
Adrian Dumitrescu\thanks{%
Department of Computer Science,
University of Wisconsin--Milwaukee, USA\@.
Email:~\texttt{dumitres@uwm.edu}}}
\begin{document}

\maketitle

\begin{abstract}
We study the selection problem, namely that of computing the $i$th
order statistic of $n$ given elements. Here we offer a data structure called
\emph{selectable sloppy heap} handling a dynamic version in which upon request:
(i)~a new element is inserted or
(ii)~an element of a prescribed quantile group is deleted from the data structure. 
Each operation is executed in (ideal!) constant time---and is thus
independent of $n$ (the number of elements stored in the data structure)---provided
that the number of quantile groups is fixed. 
This is the first result of this kind accommodating both insertion and deletion
in constant time.
As such, our data structure outperforms the soft heap data structure of Chazelle
(which only offers constant amortized complexity for a fixed error rate $0<\eps \leq 1/2$)
in applications such as dynamic percentile maintenance.
The design demonstrates how slowing down a certain computation can speed up
the data structure.


\end{abstract}

\section{Introduction} \label{sec:intro}

The following problem has been devised by Fredman about 25 years ago,
for inclusion in homework assignments for an algorithms course.
This paper both generalizes and strengthens that result.
\begin{quote}
A ``very sloppy heap'' (abbreviated {\em vsh}) is a data
structure for performing the following operations on a set $S$:
(i) {\em insert} and (ii) {\em delete-small}. The latter operation
deletes (and returns) an element $x$ which is among the $\lceil n/2 \rceil$
smallest elements in the set, where $n$ is the current size of the set. 
Explain how to implement a {\em vsh} in constant amortized time per
operation. 
\end{quote}

Together with sorting, selection is one of the most widely used
procedure in computer algorithms. Given a sequence $A$ of $n$ numbers
and an integer (selection) parameter $1\leq i\leq n$, 
the selection problem asks to find the $i$th smallest element in $A$. 
Sorting trivially solves the selection problem; however, 
a higher level of sophistication is required by a linear time algorithm.  
A now classic approach for selection~\cite{BFP+73,FR75,Hy76,SPP76,Yap76} 
from the 1970s is to use an element in $A$ as a pivot to partition $A$
into two smaller subsequences and recurse on one of them 
with a (possibly different) selection parameter~$i$.  

The time complexity of this kind of algorithms is sensitive to the pivots used.
If a good pivot is used, many elements in $A$ can be discarded, while if a bad pivot is used, 
the size of the problem may be only reduced by a constant in the worst case,
leading to a quadratic worst-case running time. But carefully choosing a good pivot can be
time consuming. Choosing the pivots randomly (and thus without much effort)
yields a well-known randomized selection algorithm with expected linear running time;
see \eg,~\cite[Ch.~9.2]{CLRS09},~\cite[Ch.~13.5]{KT06}, or~\cite[Ch.~3.4]{MU05};
however its worst case running time is quadratic in $n$. 

The first deterministic linear time selection algorithm {\sc Select} 
is due to Blum~\etal~\cite{BFP+73}; it is recursive in nature. 
By using the median of medians of small disjoint groups of the input array
(of constant size at least $5$)
good pivots that reduce the size of the problem by a
constant fraction and thereby lead to $O(n)$ time overall,
can be chosen at low cost in each recursive invocation.
More recently, suitable variants of {\sc Select} with groups of $3$ and $4$
also running in $O(n)$ time have been also put forward~\cite{CD14,Z14}.  
The selection problem, and computing the median in particular are in close relation
with the problem of finding the quantiles of a set, that we describe next. 

\paragraph{Quantiles.} The $k$th \emph{quantiles} of an $n$-element
set are the $k-1$ order statistics that divide the sorted set in $k$
equal-sized groups (to within $1$); see, \eg, \cite[p.~223]{CLRS09}.
It is known that the $k$th quantiles of a set can be computed
by a recursive algorithm running in $O(n \log{k})$ time.
Such an algorithm can be modified, if needed, so that the $k$ groups
can be also output, say, each as a linked list, within the same
overall time. For $2 \leq i \leq k-1$, the $i$th group of elements
(bounded by the $(i-1)$th and the $i$th quantile) is referred to as the $i$th
\emph{quantile group}; the first quantile group consists of the elements less or equal
to the first quantile, and the $k$th quantile group consists of the elements greater or equal
to the $(k-1)$th quantile.

\smallskip
Our main result is the following; for the optimality of the dependence in $k$,
see the first remark in Section~\ref{sec:applications}. 

\begin{theorem} \label{thm:ssh}
  For any fixed integer $k$, a data structure for dynamic sets exists
  accommodating each of the following operations in constant time
   (that depends on $k$):
(i) {\sc Insert} a new element and
(ii) {\sc Delete} (and return)  an element of the $i$th quantile group
  of the current set, where $1 \leq i \leq k$. 
  The time per operation is $O(\log{k})$, which is optimal in a comparison-based model 
  even in an amortized setting.
\end{theorem}

\paragraph{Background and related problems.}
Since the selection problem is of primary importance, 
the interest in selection algorithms has remained high ever since;
see for instance~\cite{AKSS89,BCC+00,BJ85,CM89,DHU+01,DZ96,DZ99,FR75,FG79,GKP96, 
  HS69,Ho61,J88,Ki81,Pa96,YY82,Yap76}. In particular, determining the
comparison complexity for computing various order statistics including
the median has lead to many exciting questions, some of which are still
unanswered today; in this respect, Yao's hypothesis on selection~\cite[Ch.~4]{Yao74}
remains an inspiring endeavor~\cite{DZ96,Pa96,SPP76}.
A comprehensive review of early developments in selection
is provided by Knuth~\cite{Kn98}.
Computational results on the exact comparison complexity of 
finding the $i$th smallest out of $n$ for small $i,n$,
have been obtained in~\cite{GKP96,Ho13}. 
We also refer the reader to the dedicated book chapters on selection
in~\cite{AHU83,Ba88,CLRS09,DPV08,KT06} and the more recent
articles~\cite{CD14,JMRS15,Ki13}.

The selection problem is also of interest in an online or dynamic setting,
where elements are inserted or deleted. A balanced binary search tree on $n$ distinct elements
can be augmented with a \emph{size} attribute for each node, thereby allowing the retrieval
of an element of a given rank in $O(\log{n})$ time~\cite[Ch.~14]{CLRS09}.
We say that the $i$th smallest element has \emph{rank} $i$,
where $i=1,\ldots,n$, with ties broken arbitrarily.
Further, determining the rank of an element in the data structure can also
be done in $O(\log{n})$ time. 
Consequently, a dynamic order statistic operation (inserting a new element or
deleting an element of a given rank) can be accomplished within the same time bound. 

Priority queues (\emph{heaps} in particular) are ubiquitous data structures;
a typical operation set might include {\tt Insert}, {\tt Delete-Min}, {\tt Decrease-Key}
(and perhaps also {\tt Find-Min}, {\tt Increase-Key}, or other operations). See for
instance~\cite{Chan13} and the references therein for some new variants and recent
developments. 

A variant of a priority queue that allows dynamic maintenance of percentiles 
is the \emph{soft heap} data structure due to Chazelle~\cite{Ch00a};
in addition to {\tt Create} and {\tt Meld}, the data structure accommodates
{\tt Insert}, {\tt Delete} and {\tt Findmin} ({\tt Delete $x$} removes item $x$).  
Consider a sequence of $m \geq n$ operations after its creation, that includes
$n$ {\tt Insert} operations. For any $0<\eps \leq 1/2$, 
a soft heap with error rate $\eps$ supports each operation in constant amortized time,
except for {\tt Insert} which takes $O(\log{1/\eps})$ amortized time. 
The data structure ensures that at most $\eps n$ of its items are \emph{corrupted}
(their keys have been artificially raised).
{\tt Findmin} returns the minimum current key (which might be corrupted). 
In contrast, our data structure uses $O(\log{1/\eps})$ time per operation,
where $\eps \sim 1/k$, and thereby outperforms the soft heap with respect to
the worst-case time per update (insert or delete) operation.

\paragraph{Definition.}
Let $k$ be a fixed positive integer. 
A ``selectable sloppy heap'' (abbreviated {\em ssh}) is a data
structure for performing the following operations on a set $S$ with $n$ elements:
\begin{itemize} \itemsep 0pt
\item [(a)] {\sc Insert} $x$: a new element $x$ is inserted.
\item [(b)] {\sc Delete $i$}: this deletes (and returns)
  some element $x$ which belongs to the $i$th quantile group of the current set,
where $1 \leq i \leq k$; if $n<k$, the deleted element is not subject to any requirement. 
\end{itemize}

\paragraph{Outline of the paper.}
To explain the main challenges and introduce the main ideas,
we first sketch several preliminary implementations of the data structure,
meeting suboptimal (in $k$) benchmarks:
(i)~$O(k \log{k})$ amortized time per operation for the first variant in Section~\ref{sec:prelim};
(ii)~$O(k \log{k})$ worst-case time per operation for the second variant in Section~\ref{sec:prelim};
(iii)~$O(\log{k})$ amortized time per operation for the third variant in Section~\ref{sec:prelim}.
We then further refine these methods in Section~\ref{sec:main} to obtain an optimal implementation
of the data structure running in $O(\log{k})$ worst-case time per operation. 
In particular, for constant $k$, the second variant in Section~\ref{sec:prelim} and the
(main) variant in Section~\ref{sec:main} run in $O(1)$ worst-case time per operation. 
We conclude in Section~\ref{sec:applications} with applications and technical remarks. 

\paragraph{Comments and notations.} 
Duplicate elements are easily handled by the design.
Each operation request is associated with a \emph{discrete time} step,
\ie, the $j$th operation occurs during the $j$th time step, where $j=1,2,\ldots$ 
Without affecting the results, the floor and ceiling functions are omitted
in the descriptions of the algorithms and their analyses.

Let $A$ be a set; we write $x \leq A$ if $x \leq a$ for every $a \in A$.
The size of a bucket ${\tt b}$, \ie, the number of elements in ${\tt b}$, 
is denoted by $s({\tt b})$. 
If buckets ${\tt b'}$ and ${\tt b''}$ merge into a new bucket ${\tt b}$,
written as ${\tt b'} \cup {\tt b''} \to {\tt b}$, we clearly have
$s({\tt b'}) + s({\tt b''}) = s({\tt b})$. 
Similarly, if a bucket ${\tt b}$ splits into two buckets ${\tt b'},{\tt b''}$,
written as ${\tt b} \to {\tt b'} \cup {\tt b''}$, then we have 
$s({\tt b}) = s({\tt b'}) + s({\tt b''})$ as well.

\section{Some ideas and preliminary solutions} \label{sec:prelim}

\paragraph{A first variant with $O(k \log{k})$ amortized time per operation.} 
Let $n$ denote the current number of elements in the data structure,
with the elements being stored in a linear list (or an array).
By default (if $n$ is large enough) the algorithm proceeds in \emph{phases}.

If $n < 3k$, and no phase is under way, proceed by brute force: 
for \emph{\sc Insert}, add the new element to the list;
for {\sc Delete $i$}, compute the elements in the
$i$th quantile group and delete (return) one of them arbitrarily. 
Since $k$ is constant, each operation takes $O(3k)$,
\ie, runs in constant time. 

If $n \geq 3k$, and no phase is under way, start a new phase:
reorganize the data structure by computing the $3k$th quantiles
of the set and the corresponding groups; store each group in a list
(bucket) of size $n/(3k)$.
The next $n/(3k)$ operations make up the current phase.
Process each of these $n/(3k)$ operations by using exclusively elements stored
in these $3k$ buckets.
For {\sc Insert}, add the new element to an initially empty overflow bucket
(an empty overflow bucket is created at each reorganization). 
For {\sc Delete $i$}, remove any element from the bucket
$(i-1) \frac{n}{k} + \frac{n}{3k}$, \ie, from the middle
third of the $i$th quantile group (out of the total $k$). 

The reorganization takes $O(n \log{k})$ time and is executed about every $n/(3k)$
steps, where each operation counts as one step. The resulting amortized cost
per operation is
$$ O \left( \frac{n \log{k}}{n/k} \right) = O(k \log{k}), $$
namely $O(1)$ for constant $k$. The above idea is next refined so as to obtain this
as a worst-case time bound per operation.

\paragraph{A second variant with $O(k \log{k})$ worst-case time per operation.} 
It suffices to consider the case of large $n$,
namely $n \geq 3k$. The algorithm proceeds in phases: each phase starts
with a reorganization, namely computing the $(3k)$th quantiles and the
corresponding quantile groups; see~\cite[Ch.15]{CLRS09}.
The time taken is $O(n \log{(3k)})=O(n \log{k})$. 
There are $n/(3k)$ {\sc Insert} and {\sc Delete} operations,
\ie, $n/(3k)$ steps following each reorganization
until the results of the next reorganization become available:
assume that the data structure holds $n$ items
when the current reorganization starts
($n$ is redefined at the start of each reorganization)
and there are old buckets to use until the reorganization is finalized.
That is, after that many steps new buckets (however, with an old content)
become available. Use these buckets for the next $n/(3k)$ operations,
and so on.

To transform the constant amortized time guarantee into a constant worst-case
guarantee for each operation, 
spread the execution of reorganization over the next $n/(3k)$ operations,
\ie, over the entire phase. 
The resulting time per operation is bounded from above by
$$ O\left( \frac{n \log{k}}{n/(3k)} \right) = O\left( k \log{k} \right). $$

Since {\sc Delete} is serviced from the existent buckets and
{\sc Insert} is serviced using the overflow bucket,
each operation involves inserting or deleting one element from a list;
thus for constant $k$, each operation takes $O(1)$ overall time to process. 

New buckets become available about every other $n/(3k)$ operations;
this is a rather imprecise estimate because the current number of elements, $n$, changes.
Since a reorganization is spread out over multiple steps, the result becomes available
with some delay, and moreover, its content is (partially) obsolete. 
To verify that the data structure operates as intended, one needs to
check that the rank of a deleted element belongs to the required interval
(quantile group); we omit the calculation details. The key observation is 
that any one operation can affect the rank of any element by at most $1$:
to be precise, only {\sc Insert} or {\sc Delete} of a smaller element
can increase the rank of an element by $1$ or decrease the rank of an element by $1$,
respectively.

\paragraph{A third variant with $O(\log{k})$ amortized time per operation.} 
We briefly describe an implementation achieving $O(\log{k})$ amortized time
per operation that is tailored on the idea of a B-tree;
this variant is due to Fredman~\cite{Fr16}.
Use a balanced binary search tree with $\Theta(\log{k})$ levels
storing $O(k)$ splitting keys at the leafs. Each leaf comprises a bucket 
of $\Theta(n/k)$ items, with order maintained between buckets,
but \emph{not within} buckets.  
When an insertion causes a bucket to become too large, it gets split into
two buckets by performing a median selection operation.  Small buckets (due
to deletions) get merged. A bucket of size $m$, once split, won't get split sooner
than another $\Omega(m)$ operations.

When the number of elements doubles (or halves), a tree reorganization 
is triggered that partitions the present items into $6k$ uniform sized buckets,
so that the new common bucket size $m$ is $n/(6k)$ items.
These event-triggered reorganizations ensure that buckets do not become
too small unless they are target of deletions; similarly,
buckets do not become too large unless they are target of insertions.
Since the reorganization cost is $O(n \log{k})$, and $\Omega(n)$ operation requests
take place between successive reorganizations, this scheme yields $O(\log{k})$
amortized time per operation, namely $O(1)$ for constant $k$.

\section{A variant with optimal $O(\log{k})$ worst-case time per operation} \label{sec:main}

A brief examination of the approach in the $3$rd variant
reveals two bottlenecks in achieving $O(\log{k})$ worst-case time per operation:
the median computation that comes with splitting large buckets and
the tree reorganization that occurs when $n$ doubles (or halves). 
We briefly indicate below how ideas from the $2$nd and $3$rd early variants
are refined to obtain $O(\log{k})$ worst-case time per operation;
in particular, $O(1)$ time for constant $k$. 
It is shown in the process how execution in \emph{parallel} of several sequential
procedures can lead to a \emph{speed up} of the data structure.

A balanced BST for $\Theta(k)$ keys is used, subdividing the data into $O(k)$ buckets.
A \emph{size} attribute is associated with each node reflecting the
number of elements stored in the subtree rooted at the respective node. 
Modifications in the data structure at each operation are reflected in
appropriate updates of the size attributes at the $O(\log{k})$ nodes
of the $O(1)$ search paths involved, in $O(\log{k})$ time overhead per operation. 

If $n$ is the current number of elements, each bucket holds at most $n/(3k)$ elements,
and so each of the $k$ quantile groups contains at least one of the buckets entirely.
By choosing one such bucket at the bottom of the search path in the BST for executing
a {\sc Delete} operation from the $i$th quantile group guarantees its correctness. 

The $n$ elements are kept in $\Theta(k)$ buckets of maximum size $O(n/k)$
that form the $\Theta(k)$ leaves of a balanced binary search tree $\A$.
As in the third variant (in Section~\ref{sec:prelim}), each leaf holds a bucket,
with order maintained between buckets, but \emph{not within} buckets.  
Buckets that become too large are split in order to enforce a suitable
upper limit on the bucket size. Median finding procedures along with other preparatory
and follow-up procedures accompanying bucket splits are scheduled as background computation,
as in the second variant (in Section~\ref{sec:prelim}).

Our data structure merges small buckets in order to keep the number of buckets under 
control and renounces the periodic tree reorganizations by introducing new elements of design:
a round robin process and a priority queue jointly control the maximum bucket size
and the number of buckets in the BST.
These mechanisms are introduced to prevent buckets becoming too small or too large
as an effect of changes in the total number of elements, $n$,
and \emph{not} necessarily as an effect of operations directed to them.

\paragraph{Outline and features.}
Let $N:=12k$. For illustrating the main ideas, assume now that $n \geq N$.
The buckets are linked in a doubly-linked linear list $\B$, in key order;
adding two links between the last and the first bucket yields a circular list $\C$,
referred to as the \emph{round robin} list. We note that $\B$ and $\C$ are two views of
the same data structure.

Each operation request translates to locating a suitable bucket for implementing the request.
The circular list is traversed in a round robin fashion, so that the current round robin
bucket in the list is also examined during the current operation request.
The round robin process ensures that (i)~the buckets do no exceed their maximum capacity,
and (ii)~certain ``long-term'' preparatory bucket-splitting procedures are run in the background
over a succession of non-consecutive discrete time steps allocated to the same bucket. 

Each bucket split entails a merge-test for the pair of adjacent buckets
with the minimum sum of sizes.   
The process of merging adjacent buckets in $\B$ is controlled by a priority queue in
the form a \emph{binary min-heap} $\H$. If $|\B|=t$, \ie, there are $t$ buckets,
$\H$ holds the $t-1$ sums of sizes $s({\tt b}) + s({\tt b^+})$, for buckets ${\tt b} \in \B$;
here ${\tt b^+}$ denotes the bucket that follows ${\tt b}$ in $\B$. 
A merge is made provided the minimum value at the top of $\H$ is below some threshold;
and $\A$, $\B$ and $\H$ are updated. Merging adjacent buckets ensures that the total number of 
buckets remains under control, regardless of which buckets are accessed by operation requests. 

\paragraph{Elements of the design.}
We highlight two: (i)~use of the priority queue $\H$ to keep the number of buckets
under control, and
(ii)~running the procedures involved at different rates as needed to ensure that
certain task deadlines and precedence constraints among them are met.
The data structure maintains the following two invariants:

\begin{itemize} \itemsep 0pt
\item [$\I1$] Each bucket contains between $1$ and $n/(3k)$ elements;
there is no limit from below imposed on the bucket size. 
\item [$\I2$] The number of buckets is between $3k$ and $N=12k$,
  as implied by the maximum bucket size, and a later argument based on
  the rules of operation (on merging adjacent buckets);
  see Action~2 and Lemma~\ref{lem:invariants}. 
\end{itemize}

Recall that all $\Theta(k)$ buckets are linked in a circular round robin list:
in key order, and with the last bucket linked to the first one. 
A pointer to the current \emph{round robin bucket}
(initialized arbitrarily, say, to the first bucket) is maintained.
Each \emph{operation request} advances this position in the list by one slot.
Observe that every bucket becomes the round robin bucket about every $\Theta(k)$
discrete time steps. 
Further, each operation request leads via the search path in $\A$ to one of the
buckets, referred to as the \emph{operation bucket}. 
Executing one operation is done by a sequence of actions performed in the operation bucket
and the round robin bucket; these are referred to as the \emph{current} buckets.
All these actions are therefore associated with the same discrete time step.
Every bucket update (this includes creation or deletion of buckets)
entails a corresponding update of the binary heap $\H$ in the (at most)
two heap elements $s({\tt b}) + s({\tt b^+})$ that depend on it.  

A bucket is declared \emph{large} if its current size exceeds
$9/10$ of the maximum allowed, \ie, if the current size exceeds $9n/(30k)$. 
All other buckets are declared \emph{regular}. 
Each operation may cause an update in the \emph{status} of
the round robin bucket and the operation bucket (among large and regular). 

\paragraph{Action 1.} Execute the requested operation: either add a new element to
the respective bucket, or delete an arbitrary element from it.
If the operation bucket becomes empty as a result of the current
{\sc Delete} operation, it is deleted, and the BST is correspondingly updated
in $O(\log{k})$ time. Status updates in the current two buckets are made, if needed.
If a median finding procedure or a follow up procedure is under way in the current
operation bucket or the current round robin bucket (see details below),
the next segment consisting of $500$ (or less) \emph{elementary operations}
of this procedure is executed. 

\paragraph{Action 2 (merge-test).} 
Let $\sigma:=s({\tt b}) + s({\tt b^+})$ be the minimum value at the top of the heap $\H$.
If $\sigma \leq 5n/(30k)$, merge the two buckets into one:
${\tt b} \cup {\tt b^+} \to {\tt b}$, and update the tree $\A$,
the bucket list $\B$ and the binary heap $\H$ to reflect this change:
(i)~delete the two buckets that are merged and insert the new one that results into
$\A$ and $\B$;
(ii)~extract the minimum element from $\H$ and insert the two new sum of sizes of
two adjacent buckets formed by the new bucket into $\H$ (if they exist). 

Handling of $\A,\B$ and $\H$ take $O(\log{k})$ time, $O(1)$ time, and  $O(\log{k})$ time,
respectively.  
In particular, this action ensures that the size of any new bucket obtained by merge
is at most $5n/(30k)$. It is worth noting that a bucket for which a partitioning
procedure is under way, as described below, cannot be part of a pair of buckets to merge
(\ie, passing the merge-test).

\paragraph{Action 3 (finalizing a split).}
Similar to the merge operations in Action~2,
finalizing the split can be completed in $O(\log{k})$ time:
it essentially involves updating $\A$, $\B$ and $\H$.
It will be shown subsequently that the size of any \emph{new} bucket
resulting from a split is in the range $[4n/(30k), 6n/(30k)]$.

\medskip
Besides Actions $1-3$, there are actions associated with procedures running in the background,
whose execution is spread out over a succession of non-consecutive discrete time steps
allocated to the same bucket. The procedures are in preparation of splitting large buckets. 

\paragraph{Splitting a large bucket.}
For simplicity of exposition, we assume that $n =\Omega(k)$, for a sufficiently large
constant factor.
If the current bucket ${\tt b}$ is large, \ie, $s({\tt b}) \geq 9n/(30k)$, 
and no procedure is active in the current bucket, let $n_0:=n$
(the number of elements existent when the procedure is initiated);
and place $ 8.8 n_0/(30k) $ elements into (a main part) $P_1$ and the remaining 
$s({\tt b}) - |P_1|$ elements into (a secondary part) $P_2$.
Observe that $0.2 n_0/(30k) \leq |P_2| \leq 0.4 n_0/(30k)$. 
Any insertions and deletions from the current bucket until the split is finalized
are performed using $P_2$. A \emph{balanced partition} of the current bucket
will be obtained in at most $n_0/10$ time steps. 

Consider the following $n_0/10$ operation requests. By Lemma~\ref{lem:invariants} below,
the number of buckets is at most $12k$ at any time,
and so at least $(n_0/10)/(12k)=n_0/(120k)$ time steps are allocated to ${\tt b}$
in the round robin process by the $n_0/10$ time mark.
Let $t_1$, where $n_0/(120k) \leq t_1 <n_0/10$, mark the first occurrence
when a total of $n_0/(120k)$ discrete time steps have been allocated to ${\tt b}$
(as operation steps or round robin steps). 
As shown next, this number of steps suffices for finalizing the
balanced partition and the split of ${\tt b}$. The process calls two
procedures, labeled (A) and (B): 

\smallskip
(A) Start a \emph{median finding} procedure, \ie, for finding the two quantile groups
$Q_1,Q_2$ of $P_1$: an element $m \in Q_1 \cup Q_2 =P_1$ is found so that
$Q_1 \leq m \leq Q_2$ and $||Q_1| -|Q_2|| \leq 1$. 
At the point when this procedure is launched, the computational steps
are scheduled so that every time the bucket gets accessed
(either as the operation bucket or as the round robin bucket), $500$
computational steps for selecting the median element of $P_1$ take place. 
Assume for concreteness that median finding on an input set $S$ takes at most $10|S|$ 
elementary operations; when applied to $P_1$,
we have $|P_1| \leq 8.8 n_0/(30k)$, and thus $88 n_0/(30k)$ elementary operations suffice.
At the rate of $500$ per discrete time step, it follows that
$88 n_0/(15000k) \leq n_0/(160k)$ time steps suffice for finding the median of $P_1$. 

\smallskip
(B) After the median has been determined, a \emph{follow up} procedure is initiated in the
same bucket. It aims at reducing and finally eliminating the leftover of $P_2$,
so that in another $n_0/(480k)$ discrete steps, a balanced partition of
the current bucket is obtained.
The follow up procedure starts with the two quantile groups of the same size, $4.4n_0/(30k)$,
as created by the median finding procedure, by comparing each element of $P_2$ against $m$
and properly placing it in one of the two groups. This procedure runs at the rate of $10$ items
per discrete time step accessing the current bucket
(either as the operation bucket or as the round robin bucket), 
until finally the partitioning process is completed with all elements in the current
bucket properly placed against the pivot $m$. 
Note that $|P_2| \leq 0.4 n_0/(30k) + n_0/(480) \leq 0.5 n_0/(30k)$ at any time;
at the rate of $10$ items per discrete time step, it follows that
$0.5 n_0/(300k) \leq n_0/(480k)$ time steps suffice for completing the follow up procedure.

\medskip
The two procedures terminate within a total of at most $n_0/(160k) + n_0/(480k)=n_0/(120k)$
discrete time steps, as required. The parameters are chosen so that the split takes place
before the large bucket ${\tt b}$ becomes illegal.

Recall that $0.9 n_0 \leq n \leq 1.1 n_0$; thus, in terms of the current number of elements,
$n$, the split of a large\footnote{This qualification refers to
the time when the partitioning procedure was initiated.} bucket ${\tt b}$ produces
two smaller buckets ${\tt b'}$, ${\tt b''}$, where
\begin{equation} \label{eq:split}
\frac{4n}{30k} \leq \frac{4.4 n_0}{30k} \leq s({\tt b'}), s({\tt b''}) \leq 
\frac{4.4 n_0}{30k} + \frac{0.4 n_0}{30k} + \frac{n_0}{480k}
\leq \frac{5 n_0}{30k} \leq \frac{6n}{30k}.
\end{equation}

Overlapping partitioning phases in the same bucket are precluded by the fact that no
fewer than $3n/(30k)$ data structure operations accessing a given fresh
bucket can trigger the next launching of a partitioning phase for that bucket. 

\paragraph{Remark.}
Let ${\tt b}$ be a new bucket produced in the current operation.
If ${\tt b}$ is generated by a split operation, then 
$ \frac{4n}{30k} \leq s({\tt b}) \leq \frac{6n}{30k}$ by~\eqref{eq:split}.
If  ${\tt b}$ is generated by a merge operation, then
$ s({\tt b}) \leq \frac{5n}{30k}$ by the merge-test. 

\paragraph{Analysis of merging buckets and maintaining the two invariants.}
To prove that the two invariants $\I1$ and $\I2$ are maintained, we need the following
key fact.

\begin{lemma} \label{lem:invariants}
The number of buckets is at most $N= 12k$ at any time.
\end{lemma}
\begin{proof}
  Let $t$ denote the number of buckets after the current operation is executed,
  and $j=1,2,\ldots$ denote the discrete time steps.
Let $B_1,\ldots,B_t$ be the buckets after the current step in key order.
Write $a_i= s(B_i)$, for $i=1,\ldots,t$.
We proceed by induction on $j$ and show that,
if the number of buckets in $\A$ (and $\B$) is at most $N$ after each of the
preceding $N$ time steps, it remains at most $N$ after the current time step.
Observe that the number of buckets can only increase by one after a bucket split,
and there can be at most two splits associated with a discrete time step.
The induction basis is $j \leq N$, and then indeed, we have $t \leq j \leq N$, as required.

For the induction step, assume that the number of buckets is at most $N$
after the previous operation. If no bucket splits occur during the execution of the
current operation, the number of buckets remains unchanged, and is thus 
still at most $N$ after the current operation. If bucket splits occur,
it suffices to show that the number of buckets is at most $N$ after each split. 
Consider a split operation, and let $\sigma:=s({\tt b}) + s({\tt b^+})$ be the minimum value
at the top of the heap $\H$ after the split.
There are two cases:

\emph{Case 1.} $\sigma \leq 5n/(30k)$, and thus the merge is executed.
Consequently, the number of buckets after the split is still at most $N$,
as required.

\emph{Case 2.} $\sigma> 5n/(30k)$, and thus no merge is executed.
Since $\H$ is a min-heap, we have
\begin{equation} \label{eq:adjacent}
a_i + a_{i+1} > \frac{5n}{30k}, \text{ for } i=1,\ldots,t-1.
\end{equation}
Adding these $t-1$ inequalities yields
$$ 2n= 2 \sum_{i=1}^t a_i > (t-1) \frac{5n}{30k}, $$
or $t \leq 12k=N$, as claimed, and concluding the induction step.  
\end{proof}

\paragraph{Summary of the analysis.}
As seen from the preceding paragraphs and the fact that the number of buckets is $t=O(k)$,
the total time per operation is $O(\log{t}) + O(\log{k}) + O(1) = O(\log{k})$, as required.

\section{Applications and technical remarks} \label{sec:applications}

~~~~1. An argument due to Fredman~\cite{Fr16}
shows that executing a sequence of $n$ operations
(from among {\sc Insert} and {\sc Delete $i$}, where $1 \leq i \leq k$)
requires $\Omega(n \log{k})$ time in the worst case, regardless of the
implementation of the data structure. The argument relies on the information theory
lower bound~\cite[Ch.~5.3.1]{Kn98}.
(For completeness, it is included in the Appendix.)

\smallskip
2. In addition to the two operations provided {\sc Insert} and {\sc Delete $i$},
the following operation can be also accommodated at no increase in cost:
{\sc Read $i$}: this returns some element $x$ which belongs to the $i$th
quantile group of the current set, where $1 \leq i \leq k$; the element remains
part of the data structure.

\smallskip
3. The new data structure finds applications in settings with
large dynamic sets where insertions and deletions need to be handled fast,
and there is no need to be very precise with the ranks of the elements handled.
One such application is \emph{approximate sorting}. An array of $n$ elements is said
to be $L$-\emph{sorted} if for every $i \in [n]$, we have $|{\tt rank}(a_i) - i| \leq L$;
see also~\cite{EC+W91,Ma85}. As shown by the lower bound argument, the selectable sloppy heap
can be used to $L$-sort $n$ elements in $O(n \log{(n/L)})$ time.

\smallskip
4. As mentioned in the introduction, the selection problem, and computing the median
in particular, are in close relation with the problem of finding the quantiles of a set.
A typical use of the median relies on its property of being both larger or equal
and smaller or equal than a constant fraction of the $n$ elements.
Following Yao~\cite{Yao74}, an element (think of a player, in the context of tournaments)  
is said to be $(i,j)$-\emph{mediocre} if it is neither among the top $i$ nor
among the bottom $j$ of a totally ordered set $S$. As such, an element from the
middle third quantile group of an $n$-element set is
$(\lfloor n/3 \rfloor, \lfloor n/3 \rfloor)$-mediocre; or simply \emph{mediocre}, for short.
Repeatedly returning a mediocre element in a dynamic setting (with an initially
empty set) can thus be accomplished via the new data structure,
by setting $k=3$ and then repeatedly executing {\sc Delete $2$} or {\sc Read $2$}, as needed,
at a minimal (constant) query cost.
Similarly, putting $k=100$ sets up the data structure for dynamic percentile maintenance
in $O(1)$ time per operation; to the best of our knowledge, this outperforms previously
available data structures with respect to this application; \eg, the soft heap data structure
of Chazelle~\cite{Ch00a} only offers constant amortized time per operation.

\smallskip
5. Here we continue our earlier discussion in Section~\ref{sec:intro} on using BSTs
in the context of dynamic selection taking into account the new data structure.
Traditionally, search trees do not allow duplicate keys; it is however an easy matter
to augment each node of such a tree with a \emph{multiplicity} attribute associated
with each distinct key. In particular, the multiplicity attribute is taken into account
when computing the size attribute of a node. 
Now having a balanced search tree augmented as described, insertion, deletion and search
all take $O(\log{n})$ time. 

Our construction method can be used for any parameter $k$, with $2 \leq k \leq n$.
The data structure obtained in this way can be viewed as an \emph{approximate search tree}.
In particular, one can search for a given key and the approximate rank of a search key
can be determined. To be precise, the following hold:
(1)~For any $2 \leq k \leq n$, an approximate search tree on $n$ items (duplicate keys allowed)
  can be constructed in $O(n \log{k})$ time.
(2)~Search for a given key takes $O(n/k + \log{k})$ time per operation; its
  approximate rank, \ie, the quantile group to which it belongs (out of $k$),
  can be reported in $O(\log{k})$ time (regardless of the presence of the key!). 
(3)~Insertion can be accommodated in $O(\log{k})$ time per operation.
(4)~Deletion of an unspecified key from a given quantile group (out of $k$) takes
  $O(\log{k})$ time; while deletion of a given key takes $O(n/k + \log{k})$ time.

It is worth noting that the approximate search tree we just described appears competitive
when $k$ is small and the search function is infrequently used.
Then insertion and deletion of an unspecified key from a given quantile group
(out of~$k$) takes $O(\log{k}) $ time (the deleted element is revealed after the operation);
\eg, if $k=O(\log{n})$, these operations take $O(\log{\log{n}}) $ time.
On the other hand, the approximate search tree we just described
is no real competitor for the \emph{exact} solution previously discussed; 
indeed, no choice of $k$ in our data structure would allow
an improvement in the performance of all the basic three operations. 

\smallskip
6. An alternative solution (to that outlined in Section~\ref{sec:main})
achieving $O(\log{k})$ worst-case time per operation was conceived by Fredman~\cite{Fr16}
(after the author of the current paper has communicated him the main result).
(i)~His solution avoids the need to merge small buckets (in order to keep the number of
buckets under control) by maintaining two running copies of the data structure and
performing periodic tree reorganizations that create uniform-sized buckets.
Buckets that become large are split using a mechanism similar to that devised
here---in Section~\ref{sec:main}. 
While a complete description of this solution is not publicly available,
the fact is that the decision of avoid merging small buckets comes at a high price.
The main reasons are the need to maintain multiple running copies of the data structure
(with two copies under permanent construction, etc), handling subtle consistency issues
generally hard to satisfy, and the likely reduced speed caused by the above items. 
(ii)~Another cumbersome approach claimed to be a simplification can be found
in his subsequent arXiv posting.
As such, neither alternative solution ((i) or (ii) above) matches in elegance and simplicity
the one given here. 

\smallskip
7. As mentioned earlier, the  $O(\log{1/\eps})$ amortized complexity of the soft heap
is optimal, however this does not hold for its worst-case complexity, in regard to its
{\tt Insert}, {\tt Delete} and {\tt Findmin} operations. In contrast, the
$O(\log{1/\eps})$ worst-case complexity per operation of the selectable sloppy heap is optimal
(where $\eps \sim 1/k$). 

\smallskip
8. Interestingly enough, most applications envisioned by Chazelle~\cite{Ch00a} for the soft heap
(\ie, four items from his list of five) can be dealt with using the selectable sloppy heap;
these include dynamic maintenance of percentiles, linear-time selection,
and two versions of approximate sorting.
One can observe that the complexity analysis for the soft heap, of the error rate in particular,
is more complicated than that for the selectable sloppy heap; moreover, special effort is
required to ensure that the space needed by the soft heap is linear in the number of items present.
In at least one application, approximate sorting, the analysis of doing that with a soft heap
is more involved than that for using our structure for this task. Needless to say,
the soft heap only allows deletion of small items. As such, the new data structure presented
here compares favorably to the soft heap with respect to insertion and deletion.
Moreover, the constant amortized guarantee per operation is replaced by the stronger
constant worst-case guarantee per operation in our case. 

It is worth recalling that the soft heap was designed with a specific application in mind,
minimum spanning trees. Given a connected graph $G=(V,E)$, where $|V|=n$, $|E|=m$,
with weighted edges, Chazelle~\cite{Ch00b} showed that a minimum spanning tree (MST) can be computed
by a deterministic algorithm in $O(m \alpha(m,n))$ time,
where $\alpha$ is the (slowly growing) inverse of the Ackermann function.
(A randomized linear-time algorithm was given by Karger~\etal~\cite{KKT95}.)
The question of whether the selectable sloppy heap can be used for MST computation is left open.

\section*{Appendix} \label{sec:appendix}

In this section it is shown that executing a sequence of $n$ operations
(from among {\sc Insert} and {\sc Delete $i$}, where $1 \leq i \leq k$)
may require $\Omega(n \log{k})$ time in the worst case, regardless of the
implementation of the data structure.

\paragraph{Lower bound.}
Suppose the delete operation always removes one of the $L$ smallest items.  
Then to sort $n$ items:
(a) insert them, (b) do repeated deletion, and
(c) sort the sequence of deleted items in the order created by the 
    successive deletions by executing only $n \lceil \log{L} \rceil$ comparisons.

To accomplish (a) and (b) using the data structure, call {\sc Insert} $n$ times
followed by {\sc Delete $1$} $n$ times (as we will set $L= \lceil n/k \rceil$). 
To accomplish (c), start from the right end (the last of the deleted items),
and sort from right-to-left by successive insertions into a sorted list
being formed, inserting each successive item from the left end.  A given 
item to be inserted would end up in a position no further than $L$ from the 
front of the sorted list under formation, since when it was deleted, it
was one of the $L$ smallest items.  Therefore its correct position 
can be obtained by a $\lceil \log{L} \rceil$-comparison binary search
(as in binary insertion sort~\cite[Ch.~5.2.1]{Kn98}).

Since by the information theory lower bound~\cite[Ch.~5.3.1]{Kn98}, 
(a)-(c) accomplish a task requiring $\lceil \log{n!} \rceil = n \log{n} -O(n)$ comparisons
in the worst case, and (c)~requires only $n \lceil \log{L} \rceil = n \log{L} +O(n)$ comparisons,
it must be that (a) and (b) jointly require $n \log(n/L) -O(n)$ comparisons in the worst case.
Setting $L= \lceil n/k \rceil$ yields that (a) and (b) jointly require
$n \log{k} -O(n) = \Omega(n \log{k})$ comparisons.

Finally, substitute $n$ for $2n$. It follows that there exists a sequence
of $n$ operations that require $\Omega(n \log{k})$ time, thus $\Omega(\log{k})$ amortized time
per operation.

\end{document}